\documentclass[12pt]{article}
\usepackage{amsmath, amsthm, amssymb, amstext}
\usepackage{mathrsfs}
\usepackage{array, amsfonts, mathrsfs}
\usepackage{latexsym, amsmath}
\usepackage{graphicx,color}

\newtheorem{Lemma}{Lemma}
\newtheorem{Theorem}{Theorem}
\newtheorem{Convention}{Convention}

\newtheorem{Corollary}{Corollary}
\newtheorem{Proposition}{Proposition}

\date{}

\begin{document}

\author{M.I.Belishev\thanks {St. Petersburg Department of Steklov Mathematical
        Institute, St.Petersburg, Russia, e-mail: belishev@pdmi.ras.ru. Supported
        by the RFBR grant 20-01 627-A and Volks-Wagen Foundation.},\,
        T.Sh.Khabibullin\thanks {St.Petersburg State University, e-mail: timur19983@outlook.com.}
        }

\title{Data characterization in dynamical inverse problem for the 1d wave equation with  matrix potential}

\maketitle

\begin{abstract}
The dynamical system under consideration is
\begin{align*}
&
u_{tt}-u_{xx}+Vu=0,\qquad x>0,\,\,\,t>0;\\
& u|_{t=0}=u_t|_{t=0}=0,\,\,x\geqslant 0;\quad
u|_{x=0}=f,\,\,t\geqslant 0,
\end{align*}
where $V=V(x)$ is a matrix-valued function ({\it potential});
$f=f(t)$ is an $\mathbb R^N$-valued function of time ({\it
boundary control}); $u=u^f(x,t)$ is a {\it trajectory} (an
$\mathbb R^N$-valued function of $x$ and $t$). The input/output
map of the system is a {\it response operator} $R:f\mapsto
u^f_x(0,\cdot),\,\,\,t\geqslant0$.

The {\it inverse problem} is to determine $V$ from given $R$. To
characterize its data is to provide the necessary and sufficient
conditions on $R$ that ensure its solvability.

The procedure that solves this problem has long been known and the
characterization has been announced (Avdonin and Belishev, 1996).
However, the proof was not provided and, moreover, it turned out
that the formulation must be corrected. Our paper fills this gap.
\end{abstract}

\noindent{\bf Key words:}\,\,\,1d wave equation with matrix
potential, reachable sets, controllability, propagation of
singularities, characterization of inverse data.

\noindent{\bf MSC:}\,\,\,35R30, 46-XX, 47-XX.
\bigskip

\section{Introduction}

\subsubsection*{About paper}
The subject of this work is the characterization of data in the
dynamical inverse problem for the one-dimensional vector wave
equation on semi-axis with matrix potential. To characterize the
data is to provide the necessary and sufficient conditions
ensuring the solvability of the inverse problem.

The inverse problem under consideration is to recover the matrix
potential from dynamical data (the response operator); it has long
been solved (see \cite{Avd Bel 96, Avd Bel Rozhk 98}). This is one
of the first problems solved by the BC-method. The issue is
exhausted if one needs only a procedure that determines the
potential from the data. However, in the understanding of
specialists, the inverse problem is completely solved if, in
addition to the procedure, the data characterization is provided.
If the potential is self-adjoint, the solvability conditions are
well known and, in fact, are reduced to positive definiteness of
the so-called {\it connecting operator} (CO) of the dynamical
system with boundary control, the evolution of which is governed
by the Sturm-Liouville operator with the given potential
\cite{BBlag book,Blag_71,Blag2}. There was a conjecture that, in
the general (non-self-adjoint) case, solvability is ensured by the
isomorphism of a relevant analogue of CO, and, moreover, this
result was announced in \cite{Avd Bel 96}. However, the proof was
not given and, moreover, certain doubts arose about sufficiency of
this condition. In particular, it was unclear what properties of
the CO provide the {\it locality} of the potential, i.e., the
absence of nonlocal Volterra additives in it. The question
remained open and the main purpose of our paper is to fill this
gap in the theory of one-dimensional dynamical inverse problems.

\subsubsection*{Statement and results}

All spaces, classes of functions and matrices in the paper are
real. We denote $\Omega:=[0,\infty)$ and $\Omega^T:=[0,T]\subset
\Omega$.
\smallskip

\noindent$\bullet$\,\,\,The forward problem is an initial-boundary
value problem of the form
    \begin{equation}\label{forward problem}
    \begin{cases}
    u_{tt}-u_{xx}+V\,u=0, & x>0,\ 0<t<T\\
    u|_{t=0}=u_t|_{t=0}=0,& x\geqslant 0\\
    u|_{x=0}=f, & 0\leqslant t\leqslant T,\\
    \end{cases}
    \end{equation}
where $V\in C^1_{\rm loc}(\Omega;\mathbb M^N)$ is a (real) matrix
valued function ({\it potential}), defined on the semi-axis
$x\geqslant 0$, $T>0$ the final moment of time; $f\in
L_2(([0,T];\mathbb{R}^N)$ a {\it boundary control}; $u=u^f(x,t)$
is a solution ({\it wave}) - an $\mathbb R^N$-valued function of
variables $x$ and $t$. Due to the hyperbolicity of the problem
(\ref{forward problem}), the relation ${\rm
supp\,}u^f(\cdot,t)\subset\Omega^t$ holds for all $t$.

Let $\mathscr F^T:=L_2([0,T];\mathbb R^N)$ be the space of
controls. The waves $u^f(\cdot,t)$ are time-dependent elements of
the space $\mathscr H^T:=L_2(\Omega^T;\mathbb R^N)$. Considering
the problem  (\ref{forward problem}) as a dynamical system, we
introduce a {\it control operator} $W^T:\mathscr F^T\to\mathscr
H^T$, acting by the rule:
    \begin{equation*}
(W^Tf)(x):=u^f(x,T), \quad x\in\Omega^T\,.
    \end{equation*}

Owing to the hyperbolicity of problem (\ref{forward problem}), its
extension of the form
   \begin{equation}\label{forward problem ext}
    \begin{cases}
    u_{tt}-u_{xx}+V\,u=0, & 0<x<T,\,\, 0<t<2T-x\\
    u|_{t<x}=0\\
    u|_{x=0}=f, & 0\leqslant t\leqslant 2T\\
   \end{cases}
  \end{equation}
is a well-posed initial boundary-valued problem, with which one
associates the so-called {\it extended response operator}
 \begin{equation*}
(R^{2T}f)(t)\,:=\,u^f_x(0,t), \quad 0\leqslant t\leqslant 2T,
    \end{equation*}
acting in the space $\mathscr F^{2T}$. Like all system
(\ref{forward problem}) attributes, the operator $R^{2T}$ is
determined by the potential $V|_{\Omega^T}$ (does not depend on
the values of $V|_{x>T}$).
\smallskip

\noindent$\bullet$\,\,\,The problem
    \begin{equation}\label{forward problem flat}
    \begin{cases}
    u_{tt}-u_{xx}+V_\flat\,u=0, & x>0,\ 0<t<T\\
    u|_{t=0}=u_t|_{t=0}=0,& x\geqslant 0\\
    u_{x=0}=f, & 0\leqslant t\leqslant T\\
    \end{cases}
    \end{equation}
with the potential $V_\flat(x):=(V(x))^\flat$, where $(...)^\flat:
\mathbb M^N\to\mathbb M^N$ is the matrix transposition, is said to
be {\it dual} to problem (\ref{forward problem}). Its solution
$u=u^f_\flat(x,t)$ possesses the same properties as $u^f$; the
control operator is
  \begin{equation*}
(W^T_\flat f)(x):=u^f_\flat(x,T), \quad x\in\Omega^T\,.
  \end{equation*}

The map $C^T:\mathscr F^T\to\mathscr F^T$,
    \begin{equation*}
    C^T\,:=\,(W^T_\flat)^*W^T
    \end{equation*}
is called a {\it connecting operator}. It is expressed via the
operator $R^{2T}$ by a simple and explicit relation established in
\cite{Avd Bel 96}.
\smallskip

\noindent$\bullet$\,\,\,The inverse problem is to recover the
potential $V|_{\Omega^T}$ from the given operator $R^{2T}$. Such a
{\it local} statement was originated by A.S.Blagovestchenskii in
\cite{Blag_71}; it is relevant to the hyperbolicity of the problem
(\ref{forward problem}).

The main result of the paper is as follows. Along with problems
(\ref{forward problem}) and (\ref{forward problem flat}), we
consider a family of "shortened" problems with final moments
$t=\xi\leqslant T$, each of which has its own connecting operator
$C^\xi$, acting in the corresponding space $\mathscr F^\xi$. All
$C^\xi$ are defined by $R^{2T}$. We show that $R^{2T}$ is the
response operator of a system (\ref{forward problem}) if and only
if {\it all} operators $C^\xi$ are isomorphisms. The necessity is
known: it is established in \cite{Avd Bel 96} in course of
analysis of the forward problem. The sufficiency was announced in
the same paper, but the proof still has not been provided. Our
work fills this gap. At the same time, the mistake made in
\cite{Avd Bel 96} is corrected: the assertion that for the
solvability of the inverse problem it is enough {\it only} $C^T$
to be isomorphism, turns out to be wrong. All $C^\xi$ have to be
isomorphisms.

\section{Forward problem}

\subsubsection*{Properties of waves}
Here the known properties of the solutions to problem
(\ref{forward problem}) are listed.  They are provided or easily
extracted from the results of \cite{BBlag book,Blag2}.
 \begin{Convention}\label{Conv 1}
All time-dependent functions are extended to $t<0$ by zero.
 \end{Convention}
\smallskip

\noindent$\bullet$\,\,\,Introduce the class of smooth controls
    $$
    \mathscr M^T\,:=\,\{f\in C^2([0,T];\mathbb R^N)\,|\,\,{\rm
    supp\,}f\subset(0,T]\},
    $$
which vanish near $t=0$. For $f\in \mathscr M^T$ the problem
(\ref{forward problem}) has a unique classical solution $u^f$ and
the representation
    \begin{equation}\label{solution_of_main_problem}
    u^f(x,t)=f(t-x)+\int_x^t w(x,s)f(t-s)\,ds,\qquad x\in\Omega^T,0\leqslant t\leqslant
    T
    \end{equation}
holds with the kernel $w$ that solves the Goursat matrix problem
    \begin{equation}\label{goursat_problem}
    \begin{cases}
    w_{tt}-w_{xx}+V\,w=0, & 0<x<t<T\\
    w(0,t)=0, & 0 \leqslant t \leqslant T\\
    w(x,x)=-\frac{1}{2}\int_0^x V(s)\,ds, & x\in\Omega^T
    \end{cases}
    \end{equation}
and is $C^2$-smooth in the domain
$\{(x,t)\,|\,\,x\in\Omega^T,\,\,0\leqslant x\leqslant t\leqslant
T\}$.

For $f\in\mathscr F^T:=L_2([0,T];\mathbb R^N)$ the right-hand side
of (\ref{solution_of_main_problem}) is well defined and regarded
as a (generalized) solution to the problem (\ref{forward problem})
of the class $C([0,T];L_2(\Omega^T))$. In the subsequent, we use
the following of its properties.
\smallskip

\noindent$\bf 1.$\,\,\,The relation
 \begin{equation}\label{supp u^f}
   {\rm supp\,} u^f(\cdot,t)\,\subset \Omega^t,\qquad t\geqslant 0
    \end{equation}
holds and shows that the waves propagate in the semi-axis
$x\geqslant 0$ with the speed $1$.
\smallskip

\noindent$\bf 2.$\,\,\,For the controls $f_\tau(t):=f(t-\tau)$,
which act with the delay $\tau>0$, one has
    \begin{equation}\label{shift time}
    u^{f_\tau}(\cdot,t)\,=\,u^f(\cdot,t-\tau),\qquad t\geqslant 0
    \end{equation}
(recall the Convention \ref{Conv 1}!); as a consequence, for
smooth controls the relations
    \begin{equation}\label{diff u^f}
    u^f_t = u^{\frac{df}{dt}},\qquad u^f_{tt} =
    u^{\frac{d^2f}{dt^2}}\,\overset{\rm see\,\,(\ref{forward problem})}=u^f_{xx}-Vu^f
    \end{equation}
are valid.
\smallskip

\noindent$\bf 3.$\,\,\,As it is seen from
(\ref{solution_of_main_problem}), owing to the continuity of the
integral term, the following is valid. If the control $f$ is
piecewise continuous and has a jump at the moment
$t=T-\xi$\,\,\,($0<\xi\leqslant T$), the wave $u^f(\cdot,t)$ is
also piecewise continuous and has a jump at a point $t=\xi$, and
the equality
    \begin{equation}\label{jump}
    u^f(x,T)\bigg{|}_{x=\xi-0}^{x=\xi+0}=-f(t)\bigg{|}_{t=T-\xi-0}^{t=T-\xi+0}
    \end{equation}
holds in $\mathbb R^N$. This is a simplest geometrical optics
relation: it shows that the wave discontinuity initiated by the
jump of control propagates along the semi-axis $x\geqslant 0$ with
the unit velocity, and the `amplitude' of the discontinuity
remains constant.
\smallskip

\noindent$\bf 4.$\,\,\,All the above properties and relations are
valid for the solution $ u^f_\flat$ to the dual problem
(\ref{forward problem flat}) and the solution $ u^f$ to the
extended problem (\ref{forward problem ext}).

\subsubsection*{Dynamical system}
Here problem (\ref{forward problem}) is endowed with standard
attributes of dynamical system: spaces and operators. The system
is denoted by $\alpha^T$.
\smallskip

\noindent$\bullet$\,\,\,The space of controls is $\mathscr
F^T=L_2([0,T],\mathbb{R}^N)$ with the inner product
    $$
    (f,g)_{\mathscr F^T}\,=\,\int_0^T\langle f(t), g(t)\rangle\,dt\,,
    $$
where $\langle\cdot, \cdot\rangle$ is the standard inner product
in $\mathbb R^N$, is said to be the {\it outer space} of the
system $\alpha^T$. It contains an increasing family of subspaces
    $$
    {\mathscr F}^{T,\xi}:=\{f\in \mathscr F^T\,|\,\,{\rm supp\,}f\subset [T-\xi,T]\},\qquad
    0\leqslant \xi\leqslant T
    $$
($\mathscr F^{T,0}=\{0\},\,\mathscr F^{T,T}=\mathscr F^T$) formed
by delayed controls: $T-\xi$ is the delay, $\xi$ is the action
time.

The space $\mathscr H^T:=L_2(\Omega^T,\mathbb{R}^N)$ with inner
product
    \begin{equation*}
    (u,v)_{\mathscr H^T}:=\int_{\Omega^T} \langle u(x), v(x)\rangle\, dx
    \end{equation*}
is called the {\it inner space}, the waves $u^f(\cdot,t)$ are its
elements. It contains an increasing family of subspaces
    $$
    \mathscr H^\xi:=\{y\in \mathscr H^T\,|\,\,{\rm supp\,}y\subset
    \Omega^\xi\},\qquad 0\leqslant \xi\leqslant T
    $$
($\mathscr H^{0}=\{0\}$). By (\ref{supp u^f}), we have
$u^f(\cdot,t)\in\mathscr H^\xi$ for $0\leqslant t\leqslant \xi$.
    \smallskip

\noindent$\bullet$\,\,\, The operator $W^T:\ \mathscr F^T\to
\mathscr{H}^T$,
    \begin{equation*}
    (W^Tf)(x):=u^f(x,T),\qquad x\in\Omega^T
    \end{equation*}
is said to be the {\it control operator}. By
(\ref{solution_of_main_problem}), the representation
    \begin{equation}\label{int_control}
    (W^Tf)(x)=f(T-x)+\int_x^Tw(x,s)f(T-s)ds,\qquad
    x\in\Omega^T\,.
    \end{equation}
holds. The control operator is an isomorphism of the space
$\mathscr F^T$. Indeed, the equation $W^T f=y$ is a second kind
Volterra equation, which is solvable for any $y\in\mathscr H^T$.
Moreover, (\ref{int_control}) implies
    \begin{equation}\label{controllability}
    W^T{\mathscr F}^{T,\xi}\,=\,\mathscr H^\xi\,,\qquad 0\leqslant \xi\leqslant T\,.
    \end{equation}
The second equality in (\ref{diff u^f}) can be written as:
    \begin{equation}\label{Wd^2=LW}
    W^T\,\frac{d^2}{dt^2}\,=\,-LW^T\,,
    \end{equation}
where $L=-\frac{d^2}{dx^2}+V$ is Sturm-Liouville operator, which
governs the evolution of system $\alpha^T$.
    \smallskip

\noindent$\bullet$\,\,\,The operator $R^T:\mathscr{F}^T\to
    \mathscr{F}^T,\,\,\,{\rm Dom\,}R^T=\{f\in \mathscr{F}^T\,|\,\,
    \frac{df}{dt}\in\mathscr{F}^T,\, f(0)=0\}$,
    \begin{equation*}
    (R^Tf)(t)\,:=\,u_x^f(0,t)\,,\qquad 0\leqslant t\leqslant T
    \end{equation*}
is called the {\it response operator} of the system $\alpha^T$.
Differentiation in (\ref{solution_of_main_problem}) leads to the
representation
    \begin{equation*}
    (R^Tf)(t)=-\frac{df}{dt}(t)+\int_0^tr(t-s)f(s)ds\,,\qquad 0\leqslant t\leqslant T,
    \end{equation*}
where $r(s):=w_x(0,s)$ is a $C^1$-smooth matrix-valued function
called the {\it response function}.

Also, with the system $\alpha^T$ one associates the {\it extended
response operator} $R^{2T}:\mathscr{F}^{2T}\to
\mathscr{F}^{2T},\,\,\,{\rm Dom\,}R^{2T}=\{f\in
\mathscr{F}^{2T}\,|\,\,  \frac{df}{dt}\in\mathscr{F}^{2T},\,
f(0)=0\}$,
    \begin{equation*}
    (R^{2T}f)(t)\,:=\,u_x^f(0,t)\,,\qquad 0\leqslant t\leqslant {2T},
    \end{equation*}
where $u^f$ is a solution to the extended problem (\ref{forward
problem ext}). The representation
    \begin{equation}\label{int_reaction}
    (R^{2T}f)(t)=-f'(t)+\int_0^tr(t-s)f(s)ds\,,\qquad 0\leqslant t\leqslant {2T},
    \end{equation}
holds with the matrix-valued reply function $r$. The important
fact is that the operator $R^{2T}$ and its response function
$r|_{[0,2T]}$ are determined by the potential $V|_{\Omega^T}$ (do
not depend on its values outside $\Omega^T$).
\smallskip

\noindent$\bullet$\,\,\, The problem (\ref{forward problem flat})
describes a dynamical system, which is called dual to $\alpha^T$
and denoted by $\alpha^T_\flat$. Obviously, the dual system has
the same attributes and properties as the original one. The inner
and outer spaces of both systems are the same, the corresponding
operators $W^T_\flat,\,R^T_\flat$ and $R^{2T}_\flat$ possess the
same properties and representations. As shown in the \cite{Avd Bel
96}, the response functions are related by the equality
    \begin{equation*}
    r_\flat(t)\,=\,(r(t))^\flat\,,\qquad 0\leqslant t\leqslant
    {2T}\,.
    \end{equation*}
Note that the operator $W^T_\flat$ maps $\mathscr F^T$ onto
$\mathscr H^T$ isomorphically and
 \begin{equation}\label{controllability flat}
    W^T_\flat{\mathscr F}^{T,\xi}\,=\,\mathscr H^\xi\,,\qquad 0\leqslant \xi\leqslant T
 \end{equation}
holds. Its adjoint $(W^T_\flat)^*:\mathscr H^T\to\mathscr F^T$ is
also an isomorphism.
\smallskip

\noindent$\bullet$\,\,\, The operator $C^T:\
    \mathscr{F}^T\to\mathscr{F}^T$,
    \begin{equation}\label{connecting}
    C^T:=(W_{{\flat}}^T)^*W^T
    \end{equation}
is said to be the {\it connecting operator}. By the definition, we
have
    \begin{equation*}
    (C^Tf,g)_{\mathscr{F}^T}=(W^T f,W^T_\flat
    g)_{\mathscr H^T}=(u^f(\cdot,T),u_{{\flat}}^g(\cdot,T))_{\mathscr{H}^T}\,,
    \end{equation*}
so that operator $C^T$ connects metrics of spaces $\mathscr{F}^T$
and $\mathscr{H}^T$. As a composition of two isomorphisms, it is
an isomorphism of the outer space $\mathscr F^T$.

The key fact of the BC-method as an approach to inverse problems
is a simple and explicit relation that expresses the connecting
operator via the response operator and response function. As is
shown in the \cite{Avd Bel 96}, the representation
    \begin{equation}\label{int_connecting}
    (C^Tf)(t)=f(t)+\int_0^TC^T(t,s)f(s)\,ds\,,\qquad 0\leqslant t\leqslant {T}
    \end{equation}
with the matrix kernel
    \begin{equation*}
    C^T(t,s)=\frac{1}{2}\int_{|t-s|}^{2T-t-s}r(\eta)\,d\eta\,,\qquad 0\leqslant s,t\leqslant {T}
    \end{equation*}
holds. Thus, the connecting operator is determined by the reply
function $r|_{0\leqslant t\leqslant {2T}}$.
    \smallskip

The dual system $\alpha^T_\flat$ has connecting operator
    $$
    C^T_\flat\,:=\,(W^T)^*W^T_\flat\,=\,(C^T)^*,
    $$
for which the representation (\ref{int_connecting}) is valid,
replacing the kernel $C^T(t,s)$ by
$C^T_\flat(t,s)=[C^T(t,s)]^\flat$.

\subsubsection*{Systems $\alpha^\xi$}
Consider the family of `shortened' systems
    \begin{equation*}
    \begin{cases}
    u_{tt}-u_{xx}+V(x)u=0, & x>0,\ 0<t<\xi\\
    u|_{t=0}=u_t|_{t=0}=0,& x\geqslant 0\\
    u_{x=0}=f, & 0\leqslant t\leqslant \xi
    \end{cases}
    \end{equation*}
indexed by the parameter $0<\xi\leqslant T$. Systems $\alpha^\xi$
are equipped with their spaces $\mathscr F^\xi$ and $\mathscr
H^\xi$ and operators $W^\xi,\,R^\xi,\,R^{2\xi},\,C^\xi$. Each
$C^\xi$ is an isomorphism in $\mathscr F^\xi$, the representation
(\ref{int_connecting})
    \begin{equation}\label{int_connecting xi}
    (C^\xi f)(t)=f(t)+\int_0^\xi C^\xi(t,s)f(s)\,ds\,,\qquad 0\leqslant t\leqslant {\xi}
    \end{equation}
    with matrix kernel
    \begin{equation*}
    C^\xi(t,s)=\frac{1}{2}\int_{|t-s|}^{2\xi-t-s}r(\eta)\,d\eta\,,\qquad 0\leqslant s,t\leqslant {T}
    \end{equation*}
being valid.
\smallskip

\noindent$\bullet$\,\,\, The operators $ C^\xi $ are related to
the operator $ C^T $ as follows. Recall Convention \ref{Conv 1}
and introduce the auxiliary operators
    \begin{equation*}
    {e^{T,\xi}}: \mathscr F^\xi\to\mathscr F^T,\quad ({e^{T,\xi}} f)(t):=f(t-(T-\xi)),\qquad 0\leqslant
    t\leqslant T;
    \end{equation*}
the adjoint operators are
    \begin{equation*}
    {(e^{T,\xi})^*}: \mathscr F^T\to \mathscr F^\xi,\quad (({e^{T,\xi}})^*
    f)(t):=f(t+(T-\xi))\,,\qquad 0\leqslant t\leqslant\xi\,.
    \end{equation*}
It is easy to check that
    \begin{equation*}
    {(e^{T,\xi})^*} {e^{T,\xi}}\,=\,\mathbb I_{\mathscr F^\xi};\qquad
    {e^{T,\xi}}{(e^{T,\xi})^*}=X^{T,\xi}\,,
    \end{equation*}
where $\mathbb I_{\mathscr F^\xi}$ is the unit operator, and
$X^{T,\xi}$ is the orthogonal projector in $\mathscr F^T$ onto
${\mathscr F}^{T,\xi}$, which cuts off the $\mathbb R^N$-valued
controls to the interval $[T-\xi,T]$:
    \begin{equation}\label{cutt X}
    (X^{T,\xi}f)(t):=\begin{cases}
    0, & 0\leqslant t<T-\xi\\
    f(t), & T-\xi\leqslant t\leqslant T
    \end{cases}\,.
    \end{equation}
By the use of (\ref{int_connecting}), one easily derives
    \begin{equation}\label{C^xi via C^T}
    C^\xi\,=\,{(e^{T,\xi})^*} C^T{e^{T,\xi}}\,, \qquad 0<\xi\leqslant T\,.
    \end{equation}

\noindent$\bullet$\,\,\,As easily follows from
(\ref{int_reaction}), the relationship of the extended response
ope\-rator with the `shortened' response operators is of the form
    \begin{equation*}
    R^\xi\,=\,(e^{2T,\xi})^*R^{2T} e^{2T,\xi}\,, \qquad 0<\xi\leqslant
    2T\,.
    \end{equation*}

\section{Inverse problem}\label{Sec Inverse problem}

\subsubsection*{Statement}
As noted above, the operator $R^{2T}$ is determined by the values
of the potential $V^T$ on the segment $\Omega^T$. Hence, the
relevant statement of the inverse problem, which respects such a
locality, is as follows: {\it given $R^{2T}$ to recover
$V|_{\Omega^T}$}. Also, since to give $R^{2T}$ is to know the
response matrix-function, one needs to determine $V|_{\Omega^T}$
from the given $r|_{0\leqslant t\leqslant 2T}$.

In such a statement, the problem is solved in \cite{Avd Bel 96}
and below we briefly describe a simplified version of the
procedure for solving it. The procedure is preceded with a
description of its instruments: projectors and the so-called
amplitude formula.

\subsubsection*{Projectors}

\noindent$\bullet$\,\,\, Fix a positive $\xi\leqslant T$. In the
system $\alpha^T$, the subspace
 $$
\mathscr U^\xi:= W^T{\mathscr
F}^{T,\xi}\,=\,\{u^f(\cdot,T)\,|\,\,f\in{\mathscr
F}^{T,\xi}\}\overset{\rm see\,\,(\ref{shift
time})}=\,\{u^f(\cdot,\xi)\,|\,\,f\in{\mathscr F}^{T}\}
 $$
formed by waves, is called {\it reachable} (at the moment
$t=\xi$).

The orthogonal projector in $\mathscr H^T$ onto $\mathscr U^T$ is
said to be the {\it wave projector}. By (\ref{controllability}),
it coincides with the projector in $\mathscr H^T$ onto the
subspace $\mathscr H^\xi$, which cuts off the $\mathbb R^N$-valued
functions to $\Omega^\xi$. So, we have:
    $$
    P^\xi y\,=\,\begin{cases}
    y & \text{in}\,\,\,\Omega^\xi\\
    0  & \text{in}\,\,\,\Omega^T\setminus\Omega^\xi
    \end{cases}\,, \qquad 0\leqslant \xi\leqslant {T}\,.
    $$
\smallskip

\noindent$\bullet$\,\,\,In the space $\mathscr{F}^T$, define the
operator
 \begin{equation}\label{W cal P xi= P xi W}
    \mathcal{P}^{T,\xi}\,:=\,[W^T]^{-1}P^\xi W^T\,.
 \end{equation}
Since $W^T$ acts isomorphically, whereas
$(\mathcal{P}^{T,\xi})^2=\mathcal{P}^{T,\xi}$ obviously holds, it
is a bounded projector.  Let us describe in more detail how it
acts. Begin with a general operator lemma.
\smallskip

Let $\mathscr F$ and $\mathscr H$ be the Hilbert spaces,
$\mathscr F'\subset\mathscr F$ and $\mathscr H'\subset\mathscr H$ the (closed) subspaces;
$e:\mathscr F'\to\mathscr F$ the embedding, which satisfies $e^*e=\mathbb
I_{\mathscr F'}$ and $ee^*=X$, where $X$ projects orthogonally in $\mathscr F$
onto $\mathscr F'$. Denote $\mathscr F'_\bot:=\mathscr F\ominus\mathscr F'$ and
$\mathscr H'_\bot:=\mathscr H\ominus\mathscr H'$; let $P$ be the orthogonal projector
in $\mathscr H$ onto $\mathscr H'$.

Let $W:\mathscr F\to\mathscr H$ and $V:\mathscr F\to\mathscr H$ be isomorphisms provided
$W\mathscr F'=V\mathscr F'=\mathscr H'$. Introduce the isomorphism $C:=V^*W: \mathscr F\to\mathscr F$
and the subspace $C^{-1}\mathscr F'_\bot\subset \mathscr F$. The operator
$\mathcal P:=W^{-1}PW$ acts in $\mathscr F$.
 \begin{Lemma}\label{lemma P^{T,xi}+}
Let $C':=e^*C\,e$ act isomorphically in $\mathscr F'$. Then the
decomposition in direct sum $\mathscr F=\mathscr F'\,\dot+\,C^{-1}\mathscr F'_\bot$
holds, whereas $\cal P$ is the (skew) projector in $\mathscr F$ onto
$\mathscr F'$ in parallel to $C^{-1}\mathscr F'_\bot$. The representation
  \begin{equation}\label{repres mathcal P general}
    \mathcal{P}\,=\,e\,[C']^{-1}{e^*}\,C
 \end{equation}
holds.
 \end{Lemma}
\begin{proof}
$\bf 1.$\,\,\,The operators $W$ and $V$ act isomorphically, and we
have $W\mathscr F'=V\mathscr F'=\mathscr H'$. The latter equality
implies $V^*\mathscr H'_\bot=\mathscr F'_\bot$ and leads to
 \begin{align*}
&
\mathscr F=W^{-1}[\mathscr H'\oplus\mathscr H'_\bot]=W^{-1}\mathscr H'\,\dot+\,W^{-1}\mathscr H'_\bot=\mathscr F'\,\dot+\,C^{-1}\mathscr F'_\bot\,.
 \end{align*}

$\bf 2.$\,\,\,One has $\mathcal P^2=\mathcal P$ just by the
definition of $\mathcal P$.

If $f\in\mathscr F'$ then $W f\in\mathscr{H}'$ and, hence,
 $$
\mathcal P f=W^{-1}P W f=W^{-1} W f=f\,.
 $$

If $f\in C^{-1}\mathscr F'_\bot$ then $f=C^{-1}g$ with $g\in \mathscr F'_\bot$
and one has
 $$
\mathcal{P}f=W^{-1}P W^T C^{-1}g=W^{-1}P[V^*]^{-1}g=0
 $$
in view of $[V^*]^{-1}g\in\mathscr H'_\bot$.

Thus, $\mathcal{P}$ is an idempotent, which acts identically on
$\mathscr F'$ and annuls $C^{-1}\mathscr F'_\bot$. Therefore, it
projects in $\mathscr F$ onto $\mathscr{F}'$ in parallel to
$C^{-1}\mathscr{F}'_\bot$.
\smallskip

$\bf 3.$\,\,\,Let $Q$ be the right hand side of (\ref{repres
mathcal P general}). Then we have
 $$
Q^2=e\,[C']^{-1}{e^*}\,C\,e\,[C']^{-1}{e^*}\,C=e\,[C']^{-1}C'\,[C']^{-1}{e^*}\,C=Q\,.
 $$

If $f\in\mathscr F'$ then $X f=ee^*f=f$ and, hence,
 $$
Q f=e\,[C']^{-1}{e^*}\,Cf=e\,[C']^{-1}{e^*}\,Ce\,e^*f
=e\,[C']^{-1}C'e^*f=f\,.
 $$

If $f\in C^{-1}\mathscr F'_\bot$ then $f=C^{-1}g$ with $g\in \mathscr F'_\bot$
and one has
 $$
Qf=e\,[C']^{-1}{e^*}\,CC^{-1}g=0
 $$
in view of $e^*\mathscr F'_\bot=\{0\}$.

Thus, $Q$ is an idempotent, which acts identically on $\mathscr
F'$ and annuls $C^{-1}\mathscr F'_\bot$. Therefore, it projects in
$\mathscr F$ onto $\mathscr{F}'$ in parallel to
$C^{-1}\mathscr{F}'_\bot$ and, hence, coincides with $\mathcal P$.
\end{proof}
\smallskip

Return to the definition (\ref{W cal P xi= P xi W}). Applying the
Lemma \ref{lemma P^{T,xi}+} to $\mathscr F=\mathscr F^T$, $\mathscr F'=\mathscr
F^{T,\xi}$, $\mathscr F'_\bot=\mathscr F^{T,\xi}_\bot=\mathscr F^T\ominus\mathscr
F^{T,\xi}$, $\mathscr H=\mathscr H^T$, $\mathscr H'=\mathscr H^{\xi}$, $W=W^T$,
$V=W^T_\flat$, $\mathcal P=\mathcal P^{T,\xi}$, and referring to
(\ref{C^xi via C^T}), we arrive at the following.

 \begin{Corollary}\label{Cor P^{T,xi}}
Operator $\mathcal{P}^{T,\xi}$ is the (skew) projector in
$\mathscr F^T$ onto $\mathscr{F}^{T,\xi}$ in parallel to
$[C^T]^{-1}\mathscr{F}^{T,\xi}_{\bot}$. The representation
 \begin{equation}\label{repres mathcal P}
    \mathcal{P}^{T,\xi}\,=\,{e^{T,\xi}} [C^\xi]^{-1}{(e^{T,\xi})^*} C^T\,,\qquad
    0<\xi\leqslant T
 \end{equation}
holds.
 \end{Corollary}
\noindent$\bullet$\,\,\,Due to the complete equality of systems
$\alpha^T$ and $\alpha^T_\flat$, the operator
\begin{equation}\label{W cal P xi= P xi W flat}
    \mathcal{P}^{T,\xi}_\flat\,:=\,[W^T_\flat]^{-1}P^\xi W^T_\flat
 \end{equation}
has the same properties as $\mathcal{P}^{T,\xi}$. Namely, it is
the projector in $\mathscr F^T$ onto $\mathscr F^{T,\xi}$ in parallel of the
subspace $[(C^T)^*]^{-1}\mathscr F^{T,\xi}_\bot$, and is represented in
the form.
 \begin{equation}\label{repres mathcal P flat}
    \mathcal{P}^{T,\xi}_\flat\,=\,{e^{T,\xi}} [(C^\xi)^*]^{-1}{(e^{T,\xi})^*} (C^T)^*\,,\qquad
    0<\xi\leqslant T\,.
 \end{equation}

\subsubsection*{Amplitude formula}
For a positive $\xi\leqslant T$ and control $f\in\mathscr M^T$,
one has
    \begin{align*}
    W^T\mathcal P^{T,\xi} f\overset{\rm see\,\, (\ref{W cal P xi= P xi W})}=P^\xi W^T
    f=P^\xi u^f(\cdot,T)= \begin{cases} u^f(\cdot,T) &
    \text{in}\,\,\,\Omega^\xi\\
    0 & \text{in}\,\,\,\Omega^T\setminus\Omega^\xi
    \end{cases}\,.
    \end{align*}
The control $\mathcal P^{T,\xi} f\in{\mathscr F}^{T,\xi}$ vanishes
at $0\leqslant t<T-\xi$ and, in the generic case, has an $\mathbb
R^N$-valued `jump' at the moment $t=T-\xi$. The wave $u^{\mathcal
P^{T,\xi} f}=P^\xi   u^f(\cdot,T)\in\mathscr H^\xi$ vanishes
outside $\Omega^\xi$ and also has a jump at the point $x=\xi$. The
values (amplitudes) of these jumps are related by the equality
(\ref{jump}). Since the wave $u^f(\cdot,T)$ is continuous in
$\Omega^T$, we have the relation
    \begin{equation*}
    (\mathcal{P}^{T,\xi} f)(T-\xi+0)\,=\,(P^\xi
    u^f(\cdot,T))(\xi-0)=u^f(\xi,T)\,.
    \end{equation*}
Writing it in the form
    \begin{equation}\label{AF}
    (W^T f)(\xi)\,=\,(\mathcal{P}^{T,\xi} f)(T-\xi+0)\,, \qquad
    \xi\in\Omega^T,
    \end{equation}
we get the so-called {\it amplitude formula}, which is a simplest
example of the geometrical optics relations describing the
propagation of singularities in the system $\alpha^T$. Formulas of
this type play the key role in all basic versions of the BC-method
\cite{B CIRM, B EACM, B UMN}.

\subsubsection*{Recovering the potential}
Let $r|_{0\leqslant t\leqslant 2T}$ be the response
matrix-function of a dynamical system $\alpha^T$ of the form
(\ref{forward problem}). The following procedure recovers the
potential $V$ on the segment $\Omega^T$.
\smallskip

\noindent${\bf A.}$\,\,\, Given $r$ determine the operators
$C^\xi$ for all $0\leqslant\xi\leqslant T$ by the use of
representation (\ref{int_connecting xi}).
\smallskip

\noindent${\bf B.}$\,\,\,Find the projectors $\mathcal P^{T,\xi}$
for $0\leqslant\xi\leqslant T$ by (\ref{repres mathcal P}).
\smallskip

\noindent${\bf C.}$\,\,\,Determine the control operator $W^T$ by
means of the formula (\ref{AF}) and then find its kernel $w$ (see
(\ref{int_control})). Knowing the kernel, recover the potential
$V|_{\Omega^T}$ by $V(x)=-2\,\frac{dw(x,x)}{dx}$\,\, (see
(\ref{goursat_problem})).

\section{Data characterization}

To characterize the data of the inverse problem under
consideration is to provide the necessary and sufficient
conditions on a matrix-function $r$, which guarantee that it is
the response function of some system $\alpha^T$ and thereby ensure
the solvability of the inverse problem. As will be shown, these
conditions are that all operators defined (via $r$) by the
right-hand side of (\ref{int_connecting xi}) act isomorphically in
the corresponding spaces. The necessity is already established.
Indeed, if $r$ {\it is the reply function}, the right-hand side
coincides with the connecting operator $C^\xi$ of the system
$\alpha^\xi$, which is isomorphism in $\mathscr F^\xi$.
Sufficiency is more complicated and the rest of the paper is
devoted to its proof.

The proof of sufficiency is constructive. In fact, it reduces to
applying the procedure ${\bf A\!-\!C}$ to the given function $r$.
As the result, we construct some dynamical system $\alpha^T$. In
course of the construction, it is verified that the isomorphism of
all $ C^\xi $ ensures that all steps of the procedure are
realizable. At the final step, we show that the response function
of the constructed system coincides with the function $r$, with
which we has began.

We proceed to the implementation of this program, starting with an
exact statement of the main result.
    \begin{Theorem}\label{Th 1}
A matrix-function $r \in C^1([0,2T]; \mathbb M^N)$ is the response
function of a dynamical system of the form (\ref{forward problem})
if and only if for every positive $\xi\leqslant T$ the operator
$C^\xi$ defined by
        \begin{equation}\label{int_connecting xi+}
        (C^\xi f)(t)=f(t)+\int_0^\xi C^\xi(t,s)f(t-s)\,ds\,,\qquad 0\leqslant t\leqslant {\xi}
        \end{equation}
with the kernel
        \begin{equation}\label{ker_connecting}
        C^\xi(t,s)=\frac{1}{2}\int_{|t-s|}^{2\xi-t-s}r(\eta)\,d\eta\,.
        \end{equation}
is an isomorphism of the space $\mathscr F^\xi$.
    \end{Theorem}
Let us make a remark about the notation in the forthcoming proof
of sufficiency. In the above statement, the symbol $C^\xi$ does
not assume that this operator is the connecting operator of some
system $\alpha^\xi$: it is only a candidate for this role. Other
symbols $\mathcal P^{T,\xi},\,W^T,\,R^{2T}$, etc, are also used in
this way. Such a trick simplifies the notation and, on the other
hand, clarifies the meaning of the introduced objects.

\subsubsection*{Projectors}
\noindent$\bullet$\,\,\,Realizing the plan outlined above, we
introduce the relevant analogues of the objects belonging to the
system $\alpha^T$. Begin with the projectors $\mathcal P^{T,\xi}$.
However, to introduce them by (\ref{W cal P xi= P xi W}) is not
possible because, at the moment, no $W^T$ is given. Therefore,
focusing on the representation (\ref{repres mathcal P}), we {\it
define}
 \begin{equation}\label{repres mathcal P xi+}
    \mathcal{P}^{T,\xi}\,:=\,e^{T,\xi}\, [C^\xi]^{-1} (e^{T,\xi})^* C^T\,,\qquad 0<\xi\leqslant T
 \end{equation}
that is correct since all $C^\xi$ are isomorphisms by assumption
of the Theorem. Then, by perfect analogy with (\ref{repres mathcal
P xi+}), we put
 \begin{equation}\label{repres mathcal P flat+}
    \mathcal{P}^{T,\xi}_\flat\,:=\,{e^{T,\xi}} [(C^\xi)^*]^{-1}{(e^{T,\xi})^*} (C^T)^*\,,\qquad
    0<\xi\leqslant T
 \end{equation}
(compare with (\ref{W cal P xi= P xi W flat}), (\ref{repres
mathcal P flat})) that is correct since $(C^\xi)^*$ are
isomorphisms.
   \begin{Proposition}\label{Prop CP=PflatC}
Operator $\mathcal{P}^{T,\xi}$ is the projector in $\mathscr F^T$ onto
$\mathscr F^{T,\xi}$ in parallel to $[C^T]^{-1}\mathscr F^{T,\xi}_\bot$.
Operator $\mathcal P^{T,\xi}_\flat$ projects in $\mathscr F^T$ onto
$\mathscr F^{T,\xi}$ in parallel to $[(C^T)^*]^{-1}\mathscr F^{T,\xi}_\bot$. The
equalities
 \begin{equation}\label{CP=PflatC}
        C^T\mathcal{P^\xi}=(\mathcal{P^\xi}_\flat)^*C^T,\quad
\mathcal{P^\xi}(C^T)^{-1}=(C^T)^{-1}(\mathcal{P^\xi_{{\flat}}})^*\,
  \end{equation}
are valid.
   \end{Proposition}
Indeed, to verify the first and second assertions, one needs just
to repeat the arguments of the item ${\bf 3.}$ from the proof of
Lemma \ref{lemma P^{T,xi}+}. Then the equalities (\ref{CP=PflatC})
easily follow from (\ref{repres mathcal P xi+}) and (\ref{repres
mathcal P flat+}).
\smallskip

\noindent$\bullet$\,\,\,Here we derive an efficient representation
for the projectors $\mathcal{P}^{T,\xi}$ and
$\mathcal{P}^{T,\xi}_{{\flat}}$. Recall that the projector
$X^{T,\xi}$ is defined by (\ref{cutt X}).
    \begin{Lemma}\label{L repres cal P^xi integral}
For any $0<\xi\leqslant T$, the representation
        \begin{equation}\label{P_int repres}
        (\mathcal{P}^{T,\xi}
        f)(t)=(X^{T,\xi}f)(t)+\int_0^{T-\xi}m^\xi(t,s)f(s)\,ds,\qquad
        0\leqslant t\leqslant T
        \end{equation}
holds with a piecewise $C^2$-smooth kernel $m^\xi$, which obeys
$m^\xi(t,s)|_{t>T-\xi}\equiv 0$.
    \end{Lemma}
    \begin{proof}(sketch)\,\,\,
By the general Fredholm integral equation theory, the inverse to
the isomorphism $C^\xi$ is takes the form
        \begin{equation}\label{C^xi inverse}
        ((C^\xi)^{-1}f)(t)=f(t)-\int_0^\xi l^\xi(t,s) f(s)\,ds\,,\qquad
        0\leqslant t\leqslant \xi
        \end{equation}
with a matrix kernel $l^\xi$ of the same smoothness as the kernel
$C^\xi$: it is continuous in square $[0,\xi]\times[0,\xi]$ twice
continuously differentiable outside the diagonal $t=s$.

Substituting (\ref{int_connecting xi+}) (with $\xi=T$) and
(\ref{C^xi inverse}) to the right of the definition (\ref{repres
mathcal P xi+}), as a result of cumbersome calculations
(integration by parts, changing the order of integration, etc.),
we arrive at (\ref{P_int repres}). Also, the calculations provide
 \begin{equation}\label{kernel m}
            m^\xi(t,s)=X^{T,\xi}C^T(t,s)-\int_0^\xi
            l^\xi(t-(T-\xi),\eta)\,C^T(\eta+(T-\xi),s)\,d\eta\,,
 \end{equation}
where $C^T(\cdot,\cdot)$ is the kernel (\ref{ker_connecting}) for
$\xi=T$.
\end{proof}

Note in addition that the integration limits in (\ref{P_int
repres}) correspond to how the projector $\mathcal P^{T,\xi}$
acts. If $f\in{\mathscr F}^{T,\xi}$, then $f|_{0\leqslant
t\leqslant T-\xi}=0$; therefore, the integral vanishes, which
provides $\mathcal P^{T,\xi} f={X^{T,\xi}} f=f$. Also, since
$\mathcal P^{T,\xi} f|_{0\leqslant t\leqslant T-\xi}=0$ for any
$f$, the kernel $m^\xi$ must vanish identically for $t>T-\xi$.
\smallskip

By the use of (\ref{repres mathcal P flat+}), quite analogous
arguments lead to the representation
     \begin{equation*}
    (\mathcal{P}^{T,\xi}_{{\flat}}
    f)(t)=(X^{T,\xi}f)(t)+\int_0^{T-\xi}
    m^\xi_\flat(t,s)f(s)\,ds,\qquad
    0\leqslant t\leqslant T
    \end{equation*}
with kernel $m^\xi_\flat$ having the same properties as $m^\xi$.

\subsubsection*{Operators $W^T$ and $W^T_{{\flat}}$}

\noindent$\bullet$\,\,\,The next definition is motivated by the
amplitude formula (\ref{AF}). Let us recall that
    $\mathscr H^T=L_2(\Omega^T;\mathbb R^N)$ and introduce the operator
    $W^T:\mathscr{F}^T\to \mathscr H^T$,
    \begin{equation}\label{W_T!}
    (W^Tf)(x):=(\mathcal{P}^{T,x} f)(T-x+0),\qquad
    x\in\Omega^T.
    \end{equation}
Note that, at the moment, it is just an operator defined (in
several steps) by the given function $r|_{[0,2T]}$. However, later
on, it will turn out to be the control operator of some system
$\alpha^T$. The construction of this system is the main storyline
of the proof of sufficiency.
\smallskip

To obtain a representation of $W^T$ it suffices to put $t=T-x+0$
in (\ref{P_int repres}) and take into account the form of the
kernel $m^\xi$ in (\ref{kernel m}). As a result of simple
calculations, one gets
    \begin{equation}\label{W_T_int}
    (W^Tf)(x)=f(T-x)+\int_x^Tw(x,s)f(T-s)\,ds,\qquad x\in\Omega^T
    \end{equation}
with $C^2$-smooth matrix kernel of the form
    \begin{align}
    \notag &  w(x,s) \,=\,C^T(T-x,T-s)\,-\\
    \label{W_T_int_ker} &  \,-\int_0^x
    l^x(0,\eta)\,C^T(\eta+(T-x),T-s)\, d\eta\,,\qquad 0\leqslant x\leqslant s\leqslant
    T\,,
    \end{align}
where $l^x$ is taken from (\ref{C^xi inverse}). Putting $x=0$ in
(\ref{W_T_int_ker}) and taking into account
$C^T(T,T)\overset{\rm see\,\,(\ref{ker_connecting})}=0$, we get
    \begin{equation*}
    w(0,s)\,=\,0,\qquad  0\leqslant s\leqslant T\,.
    \end{equation*}
It is recommended to compare (\ref{W_T_int}) with
(\ref{int_control}).
\smallskip

As it easily follows from (\ref{W_T_int}), $W^T$ is an isomorphism
from $\mathscr F^T$ on $\mathscr H^T$ and, moreover, the relation
 \begin{equation}\label{W^T FTxi=HTxi}
    W^T{\mathscr F}^{T,\xi}\,=\,\mathscr H^\xi\,,\qquad 0\leqslant \xi\leqslant T
    \end{equation}
holds. Also, by standard arguments of the theory of 2nd-order
Volterra integral equations, for $y\in\mathscr H^T$ one has
   \begin{equation}\label{W_T-1}
    ([W^T]^{-1}y)(t)=y(T-t)-\int_0^tw^{-1}(t,s)\,y(T-s)\,ds,\qquad 0\leqslant t\leqslant T
    \end{equation}
with kernel $w^{-1}$, twice continuously differentiable for $0
\leqslant s < t \leqslant T$ obeying
    \begin{equation*}
    w^{-1}(0,s)\,=\,0,\qquad  0\leqslant s\leqslant T\,.
    \end{equation*}
\smallskip

\noindent$\bullet$\,\,\,Quite analogously,  the operator
$W^T_\flat:\mathscr{F}^T\to \mathscr H^T$,
    \begin{equation*}
    (W^T_\flat)(x):=(\mathcal{P}^{T,x}_\flat f)(T-x+0),\qquad
    x\in\Omega^T.
    \end{equation*}
possesses the same properties as $W^T$. Namely, the
representation:
    \begin{equation}\label{W_T_int flat}
    (W^T_\flat)(x)=f(T-x)+\int_x^T w_\flat(x,s)f(T-s)\,ds,\qquad x\in\Omega^T
    \end{equation}
holds with a kernel $w_\flat$, which is of the same smoothness as
$w$ and obeys $ w_\flat(0,s)\,=\,0,\quad  0\leqslant s\leqslant
T$. It is an isomorphism, which provides
 \begin{equation}\label{W flat F txi=H xi}
W^T_\flat{\mathscr F}^{T,\xi}=\mathscr H^\xi,\quad 0\leqslant
\xi\leqslant T\,.
 \end{equation}
Its inverse has the form
 \begin{equation*}
    ([W^T_\flat]^{-1}y)(t)=y(T-t)-\int_0^tw^{-1}_\flat(t,s)\,y(T-s)\,ds,\qquad 0\leqslant t\leqslant T
    \end{equation*}
with a kernel $w^{-1}_\flat$ obeying $w^{-1}_\flat(0,s)=0,\quad
0\leqslant s\leqslant T$.
\smallskip

\noindent$\bullet$\,\,\, Let $Y^{\xi}$ be (orthogonal) projector
in $\mathscr H^T$ on $\mathscr H^\xi$, which cuts off $\mathbb
R^N$-valued functions onto the segment $\Omega^\xi$:
    \begin{equation*}
    Y^\xi y=
    \begin{cases}
    y, & \text{in}\,\,\Omega^\xi\\
    0,    & \text{in}\,\,\Omega^T\setminus\Omega^\xi
    \end{cases}\,\,.
    \end{equation*}
    \begin{Lemma}\label{lemma 3}
For every $0<\xi\leqslant T$ the relation
        \begin{equation}\label{plexus}
        W^T\mathcal{P^\xi}=Y^\xi W^T; \qquad
        W^T_{{\flat}}\mathcal{P^\xi}_\flat=Y^\xi W^T_{{\flat}}.
        \end{equation}
holds.
    \end{Lemma}
    \begin{proof} We use two facts:

1) since ${\mathscr F}^{T, \xi}\subset{\mathscr F}^{T, \xi'}$ for
$\xi<\xi'$, the projectors satisfy ${\mathcal P}^{T,
\xi}<{\mathcal P}^{T, \xi'}$, which implies
        $$
        \mathcal P^{T,x}\,\mathcal P^{T,\xi}\,=\,
        \begin{cases}
        \mathcal P^{T,x} & \text{for}\,\,\,x<\xi\\
        \mathcal P^{T,\xi} & \text{for}\,\,\,x>\xi
        \end{cases}\,;
        $$

2) if $f\in\mathscr F^T$, then $\mathcal P^{T,\xi} f\in{\mathscr
F}^{T,\xi}$; hence  ${\rm supp\,}\mathcal P^{T,\xi}
f\subset[T-\xi,T]$ and, therefore, $(\mathcal P^{T,\xi}
f)(T-x-0)=0$ for $x>\xi$.
        \smallskip

As a consequence, according to the definition (\ref{W_T!}), we
have
        \begin{align*}
        & \left(W^T\mathcal P^{T,\xi} f\right)(x)\,=\\
        & =\,    \begin{cases}
        (\mathcal P^{T,x} \mathcal P^{T,\xi} f)(T-x-0)=(\mathcal P^{T,x} f)(T-x-0)=\left(W^T f\right)(x), & x<\xi\\
        (\mathcal P^{T,x} \mathcal P^{T,\xi} f)(T-x-0)=(\mathcal P^{T,\xi}
        f)(T-x-0)=0,
        & x>\xi
        \end{cases}\,\,=\\
        & =\,(Y^\xi W^T f)(x)\,.
        \end{align*}

The second equality in (\ref{plexus}) is proved in the same way.
    \end{proof}
    \smallskip

\noindent$\bullet$\,\,\,Here a relation, which connects the
operators $W^T,\,W^T_\flat$ and $C^T$, is established. In its
form, it duplicates the definition (\ref{connecting}). However, at
the moment, $W^T$ and $W^T_\flat$ are just some operators
constructed via the function $r$ and we do not claim that $C^T$ is
the connecting operator of some system $\alpha^T$. This remains to
be proved.
    \begin{Lemma}\label{lemma C=WW}
        The relation
        \begin{equation}\label{C_and_W}
        C^T\,=\,(W^T_{{\flat}})^* W^T\,
        \end{equation}
        holds.
    \end{Lemma}
    \begin{proof}
        ${\bf 1.}$\,\,\,Let us denote $A:=W^T[C^T]^{-1}(W^T_\flat)^*$ and verify equality
        \begin{equation}\label{AY=YA}
        AY^\xi\,=\,Y^\xi A\,.
        \end{equation}

Multiplying the first equality in (\ref{plexus}) on the right by
$[C^T]^{-1}(W^T_\flat)^*$ we have
        \begin{equation}\label{111}
        W^T\mathcal P^\xi\,\,[C^T]^{-1}(W^T_\flat)^*=Y^\xi
        W^T\,\,[C^T]^{-1}(W^T_\flat)^*=Y^\xi A
        \,.
        \end{equation}
In the second equality (\ref{plexus}), passing to the adjoint
operators, one has
$(\mathcal{P^\xi}_\flat)^*(W^T_{{\flat}})^*=(W^T_{{\flat}})^*Y^\xi$.
Multiplying both parts on the left by $W^T[C^T]^{-1}$, we obtain:
 \begin{equation}\label{222}
        W^T[C^T]^{-1}(\mathcal{P^\xi}_\flat)^*(W^T_{{\flat}})^*=W^T[C^T]^{-1}(W^T_{{\flat}})^*Y^\xi=AY^\xi
        \,.
 \end{equation}
In this way, we get
$$
Y^\xi A\overset{(\ref{111})}=W^T\mathcal
P^\xi\,\,[C^T]^{-1}(W^T_\flat)^*\overset{(\ref{plexus})}=
W^T[C^T]^{-1}(\mathcal{P^\xi}_\flat)^*(W^T_{{\flat}})^*\overset{(\ref{222})}=AY^\xi\,.
$$
So (\ref{AY=YA}) is valid.
\smallskip

${\bf 2.}$\,\,\,Representations (\ref{C^xi inverse}) (with
$\xi=T$), (\ref{W_T_int}) and (\ref{W_T_int flat}) easily imply
that the operator $A=W^T\,\,[C^T]^{-1}(W^T_\flat)^*$ has the form
        \begin{equation*}
        A=\mathbb I+K\,,
        \end{equation*}
where $K$ is a compact integral operator in $\mathscr H^T$. The
commutation (\ref{AY=YA}) leads to $KY^\xi=Y^\xi K$ and, then, to
$K^*Y^\xi=Y^\xi K^*$. As a result, we have
 $$
K^*KY^\xi=K^*Y^\xi K=Y^\xi K^*K,
 $$
so that a {\it self-adjoint} operator $K^*K$ commutes with the
family of projectors $\{Y^\xi\}_{0\leqslant \xi\leqslant T}$. By
the well-known arguments of the spectral theory, the latter is
possible if and only if $K^*K$ is the multiplication by a bounded
positive measurable matrix-function. However, since $K^*K$ is {\it
compact}, this is possible if and only if $K^*K=\mathbb O$, which
is equivalent to $K=\mathbb O$. Thus, we arrive at $A=\mathbb I$.
\smallskip

${\bf 3.}$\,\,\, By the definition of $A$, we have
        $$
        A\,=\,W^T[C^T]^{-1}(W^T_\flat)^*\,=\,\mathbb I\,.
        $$
Since $W^T$ and $W^T_\flat$ are isomorphisms, the latter leads to
(\ref{C_and_W}).
    \end{proof}

\subsubsection*{Operator $L$}

Recall that $\mathscr M^T\subset \mathscr F^T$ is the class of
$C^2$-smooth controls vanishing near $t=0$. Also, note that the
set $\mathscr M^T\cap \mathscr F^{T,\xi}$ is dense in $\mathscr
F^{T,\xi}$ for all positive $\xi\leqslant T$.

Focusing on (\ref{Wd^2=LW}), let us define operator in $\mathscr
H^T$:
    \begin{equation}\label{def L}
    L:=-W^T\,\frac{d^2}{dt^2}\,[W^T]^{-1}\,,\quad {\rm Dom\,}
    L=W^T\mathscr M^T\,.
    \end{equation}
Using representation (\ref{W_T_int}) of $W^T$ and smoothness of
$w$, it is easy to make sure that ${\rm Dom\,} L=W^T\mathscr
M^T=\{y\in C^2(\Omega^T;\mathbb R^N)\,|\,\,{\rm
supp\,}y\subset[0,T)\}$.

The following result shows that $L$ is a Sturm-Liouville operator.
    \begin{Lemma}\label{lemma L Sturm-Liouville}
The representation
\begin{equation}\label{L Sturm-Liouv}
L\,=\,-\,\frac{d^2}{dx^2}\,+\,V(x)
\end{equation}
holds with potential $V(x):=-2\,\frac{dw(x,x)}{dx}\in
C^1(\Omega^T; \mathbb M^N)$, where $w$ is the kernel of the
integral part of $W^T$ in (\ref{W_T_int}).
    \end{Lemma}
    \begin{proof}(sketch)

\noindent{$\bf 1$.}\,\,\, At first, an auxiliary relation is
derived. To simplify the notation, we use
$\dot{(...)}=\frac{d}{dt}$.

Let controls $f,g\in\mathscr M^T$ obey $f(T)=g(T)=0$. Then the
following equality holds:
        \begin{equation}\label{aux rel}
        (C^T\ddot{g},f)_{\mathscr F^T}\,=\,(C^Tg,\ddot{f})_{\mathscr F^T}\,.
        \end{equation}
It can be verified by integration by parts, with regard to the
boundary conditions imposed of the controls and a specific form
(\ref{ker_connecting}) (with $\xi=T$) of kernel of $C^T$.
        \smallskip

\noindent{$\bf 2$.}\,\,\, Show that the operator $L$ is {\it
local}, i.e. satisfies ${\rm supp\,}Ly\subset{\rm supp\,}y$.

Let $y\in{\rm Dom\,}L$ and ${\rm supp\,}y\subset\Omega^\xi$. Then,
due to (\ref{W^T FTxi=HTxi}), for $f:=[W^T]^{-1}y$ we have
$f\in\mathscr M^T\cap{\mathscr F}^{T,\xi}$ and, consequently,
$\ddot f\in{\mathscr F}^{T,\xi}$. Referring again to (\ref{W^T
FTxi=HTxi}), we get: $Ly=-W^T\ddot f\in\mathscr H^\xi$, which
means ${\rm supp\,}Ly\subset\Omega^\xi$. Thus, $ L $ does not
extend the support of functions to the {\it right}.

Let $y\in{\rm Dom\,}L$ and ${\rm
supp\,}y\subset\overline{\Omega^T\setminus\Omega^\xi}=[T-\xi,T]$.
Then for $f:=[W^T]^{-1}y\in\mathscr M^T$ we obtain $f(T)=y(0)=0$.
Let $g\in\mathscr M^T\cap{\mathscr F}^{T,\xi}$ and $g(T)=0$. Thus,
$f$ and $g$ satisfy the conditions, which provide (\ref{aux rel}).
Then we have:
        \begin{align*}
        & -(Ly,W^T_\flat g)_{\mathscr H^T}=-(LW^Tf,W^T_\flat g)_{\mathscr H^T}=(W^T\ddot
        f,W^T_\flat g)_{\mathscr H^T}\overset{(\ref{C_and_W})}=(C^T\ddot
        f,g)_{\mathscr F^T}\overset{(\ref{aux rel})}=\\
        & = (C^Tf,\ddot g)_{\mathscr F^T}= (y,W^T_\flat\ddot
        g)_{\mathscr H^T}=0\,,
        \end{align*}
because $\ddot g\in{\mathscr F}^{T,\xi}$, and therefore
$W^T_\flat\ddot g\in\mathscr H^\xi$, whereas $y\in\mathscr
H^T\ominus\mathscr H^\xi$ by assumption on its support. In the
meantime, the set $\{g\in{\mathscr F}^{T,\xi}\,|\,\,g(T)=0\}$ is
dense in ${\mathscr F}^{T,\xi}$. Owing to (\ref{W flat F txi=H
xi}), the images $W^T_\flat g$ constitute a dense set in $\mathscr
H^\xi$. Therefore, the established equality $(Ly,W^T_\flat
g)_{\mathscr H^T}=0$ implies $Ly\in\mathscr H^T\ominus\mathscr
H^\xi$. The latter is equivalent to ${\rm
supp\,}Ly\subset\overline{\Omega^T\setminus\Omega^\xi}$. As a
result, $L$ does not extend the support of functions  to the {\it
left}.

So, $L$ does not extend the support of functions, i.e., acts
locally.
\smallskip

\noindent{$\bf 3$.}\,\,\,Let us show that (\ref{L Sturm-Liouv})
does hold. By (\ref{W_T_int}), for $f\in\mathscr M^T$ one easily
derives
        \begin{align}
\notag  & (W^Tf)''(x)-(W^T\ddot
        f)(x)=\\
\notag  & =V(x)f(T-x)+\int_x^T\left[w_{xx}(x,s)-w_{ss}(x,s)\right]f(T-s)\,ds=\\
\label{calculus}        & =
        V(x)(W^Tf)(x)+\int_x^T\left[w_{xx}(x,s)-w_{ss}(x,s)-V(x)w(x,s)\right]f(T-s)\,ds
        \end{align}
where $(...)'=\frac{d}{dx}$ and $V:=-2\,\frac{dw(x,x)}{dx}$. Note
that all operations, which are applied in course of the derivation
(differentiation of integrals, integration by parts, etc) are
justified owing to $C^2$-smoothness of the kernels of the
integrals under consideration. Next, substituting $f=[W^T]^{-1}y$
and using (\ref{W_T-1}), the derived relation is transformed to
        \begin{equation}\label{L local}
        (Ly)(x)=-y''(x)+V(x)y(x)+\int_x^Tk(x,s)y(s)\,ds,\qquad
        x\in\Omega^T
        \end{equation}
with a continuous kernel $k$.
\smallskip

\noindent{$\bf 4$.}\,\,\,We omit a simple proof of the following
fact: an operator of the form (\ref{L local}) is {\it local} if
and only if the integral summand is absent. Thus, we arrive at
(\ref{L Sturm-Liouv}).
  \end{proof}

\subsubsection*{Operator $L_\flat$}

\noindent$\bullet$\,\,\,By the use of the same scheme, it is
established that the operator
    \begin{equation*}
    L_\flat:=-W^T_\flat\,\frac{d^2}{dt^2}\,[W^T_\flat]^{-1}\,,\quad
    {\rm Dom\,} L_\flat=W^T_\flat\mathscr M^T
    \end{equation*}
is of the form
    \begin{equation*}
    L_\flat\,=\,-\,\frac{d^2}{dx^2}\,+\,V_\flat(x)
    \end{equation*}
with potential $V_\flat(x):=-2\,\frac{dw_\flat(x,x)}{dx}\in
C^1(\Omega^T; \mathbb M^N)$, where $w_\flat$ is the kernel of the
integral part of $W^T_\flat$ in (\ref{W_T_int flat}).
    \smallskip

    \noindent$\bullet$\,\,\,
Let us show that operators $L$ and $L_\flat$ are adjoint by
d'Alembert, i.e., for $y,v\in C^\infty_0(\Omega^T;\mathbb R^N)$
the equality
    \begin{equation*}
    (Ly,v)_{\mathscr H^T}\,=\,(y,L_\flat v)_{\mathscr H^T}\,.
    \end{equation*}
is valid. Indeed, by the choice of $y$ and $v$, the controls
$f=[W^T]^{-1}y$ and $g=[W^T_\flat]^{-1}v$ vanish at $t=0$ and obey
$f(0)=g(0)=0$. Hence, we have:
    \begin{align*}
    & (Ly,v)_{\mathscr H^T}=(W^T\ddot f,W^T_\flat g)_{\mathscr H^T}=(C^T\ddot
    f,g)_{\mathscr F^T}\overset{(\ref{aux rel})}
    = (C^Tf,\ddot g)_{\mathscr F^T}=\\
    & =(W^T f,W^T_\flat \ddot g)_{\mathscr H^T}=(y,L_\flat v)_{\mathscr H^T}\,.
    \end{align*}
As a consequence, one easily concludes that the potentials are
connected by the equality
    \begin{equation*}
    V_\flat(x)=V^\flat(x),\qquad x\in\Omega^T\,.
    \end{equation*}

\subsubsection*{Completion of the proof of Theorem \ref{Th 1}}

The operator (\ref{L Sturm-Liouv}) determines a dynamical system
$\alpha^T$ of the form
    \begin{equation}\label{forward problem+}
    \begin{cases}
    u_{tt}+Lu=0, & x>0,\ 0<t<T\\
    u|_{t=0}=u_t|_{t=0}=0,& x\geqslant 0\\
    u_{x=0}=f, & 0\leqslant t\leqslant T.\\
    \end{cases}
    \end{equation}
As is seen from (\ref{def L}), the operator $W^T$ is the control
operator of this system.

Quite analogously, system $\alpha^T_\flat$ of the form
    \begin{equation*}
    \begin{cases}
    u_{tt}+L_\flat u=0, & x>0,\ 0<t<T\\
    u|_{t=0}=u_t|_{t=0}=0,& x\geqslant 0\\
    u_{x=0}=f, & 0\leqslant t\leqslant T\\
    \end{cases}
    \end{equation*}
is controlled by the operator $W^T_\flat$.
\smallskip

As it follows from (\ref{C_and_W}), the connecting operator of the
system (\ref{forward problem+}) defined by (\ref{connecting}) {\it
coincides} with the operator $C^T$ introduced by
(\ref{int_connecting xi+}) (for $\xi=T$). Therefore, the integral
parts of these operators also coincide. The latter obviously
implies that the response matrix-function of the system
(\ref{forward problem+}) is identical to the function
$r|_{0\leqslant t\leqslant 2T}$, with which our considerations
have started.

Thus, $r|_{0\leqslant t\leqslant 2T}$ is the response function of
a dynamical system of the form (\ref{forward problem}). {\it The
sufficiency of the conditions of Theorem \ref{Th 1} is proved.}

\subsubsection*{Comments}

\noindent$\bullet$\,\,\, The deep connection between inverse
problems and the problem of triangular factorization of operators
is well known. It can be also traced in this work.

Recall the definitions. Let a monotone family ({\it nest}) of
subspaces $\mathfrak f=\{\mathscr F^\xi\}_{0\leqslant \xi\leqslant
T}:$ $\mathscr F^\xi\subset\mathscr F^{\xi'}$ for $\xi<\xi'$ be
given in a Hilbert space $\mathscr F$. An operator $Z$ is called
{\it triangular} with respect to this nest if $Z\mathscr
F^\xi\subset\mathscr F^\xi$, i.e. all the subspaces $\mathscr
F^\xi$ are invariant with respect to $Z$. We say that operators
$Z$ and $Z_\flat$, which are triangular with respect to $\mathfrak
f$, provide {\it triangular factorization} of an operator $C$ if
$C=Z_\flat^*Z$ holds.

In our paper, in the space $\mathscr F^T$, there is the nest of
the subspaces $\mathfrak f=\{{\mathscr F}^{T,\xi}\}_{0\leqslant
\xi\leqslant T}$. Let us introduce an isometry
    $$
I^T: \mathscr H^T\to\mathscr F^T, \quad (I^Ty)(t):=y(T-t), \qquad
0\leqslant t\leqslant T
    $$
and note that $(I^T)^*I^T=\mathbb I_{\mathscr F^T}$. Following
from (\ref{controllability}) and (\ref{controllability flat}),
operators $Z^T:=I^TW^T$ and $Z^T_\flat:=I^TW^T_\flat$ are
triangular with respect to $\mathfrak f$. According to
(\ref{connecting}), we have:
   \begin{equation}\label{C^T factor}
C^T=(W^T_\flat)^*W^T=(I^TW^T_\flat)^*I^TW^T\,=\,(Z^T_\flat)^*Z^T\,.
   \end{equation}
Consequently, the pair $Z^T,\,Z^T_\flat$ provides  triangular
factorization of the connecting operator of the system $\alpha^T$
with respect to the nest of subspaces, formed by delayed controls.

Thus, solving the inverse problem by the procedure ${\bf
A.}\!-\!{\bf C.}$ described in the end of the section \ref{Sec
Inverse problem}, we solve the triangular factorization problem
for the operator $C^T$ by (\ref{C^T factor}). These problems are
{\it equivalent}.

The general factorization problem for operators of the form
$\mathbb I+\text{compact}$\, is solved in \cite{GohbKrein}: see
Theorem 2.1, which provides the necessary and sufficient
conditions for its solvability. The conditions on the family of
operators $C^\xi$ adopted in our Theorem \ref{Th 1}, are quite
adequate to the mentioned classical ones.

Let us return to the question raised in the Introduction: why does
the operator $L$ given by (\ref{def L}) turn out to be {\it
local}? The explanation is in a very specific form of the kernel
of operator $C^T$: see (\ref{ker_connecting}). Such a specifics is
used, in particular, in the calculations (\ref{aux rel}) and
(\ref{calculus}).
\smallskip

\noindent$\bullet$\,\,\, A substantial difference, which
distinguishes the problem with a non-self-adjoint potential $V$
from the problem with $V^*=V$, is as follows. In the second case,
the connecting operator $C^T=(W^T)^*W^T$ is {\it positive
definite} and, hence, all the shortened operators $C^\xi$ turn out
to be such. Therefore, for characterization it suffices to require
only $C^T$ to be isomorphism: this implies isomorphism of all
$C^\xi$. In the general case, isomorphism of $C^T$ does not ensure
isomorphism of $C^\xi$. This is the mistake made in the statement
of the conditions of Theorem 3.2 in \cite{Avd Bel 96}.

There is a case when the isomorphism of all $C^\xi$ certainly
takes place. If $T>0$ is small, then the integral parts of the
operators $C^\xi$ have a small norm and all $C^\xi$ turn out to be
isomorphisms. With this reservation, the statement of Theorem 3.2
becomes true.
\smallskip

\noindent$\bullet$\,\,\,The scheme, which provides the data
characterization in this work, is traditional for the BC-method.
Its core is, first, to elaborate an efficient procedure, which
{\it solves} the inverse problem and, then, to provide the
conditions, which ensure its realizability. In one-dimensional
problems such a scheme works quite successfully: see \cite{B Pest,
BMikh Dirac IP}. There are certain results on the multidimensional
case but the list of the characteristic conditions turns out to be
rather long \cite{BV Rendiconti}.


\begin{thebibliography}{9}

\bibitem{Avd Bel 96}
S.A.Avdonin, M.I.Belishev.
\newblock {Boundary control and dynamical inverse
problem for nonselfadjoint Sturm-Liouville operator (BC-method).}
\newblock {\em  Control and Cybernetics}, 25 (1996), No 3, 429--440.

\bibitem{Avd Bel Rozhk 98}
S.A.Avdonin, M.I.Belishev, Yu.S.Ryzhkov.
\newblock {Dynamical inverse problem for nonselfadjoint Sturm-Liouville operator.}
\newblock {\em Zapiski Nauch. Semin. POMI}, 250 (1998), 7--21 (in Russian).

\bibitem{B CIRM}
M.I.Belishev.
\newblock {Boundary Control Method in Dynamical Inverse Problems - An
Introductory Course by M.I.Belishev}.
\newblock {\em Gladwell G.M.L., Morassi
A., Editors (2011). Dynamical Inverse Problems: Theory and
Application. CISM Courses and Lectures}, Vol. 529, Wien, Springer,
p. 85--150.

\bibitem{B Pest}
M.I.Belishev and A.L.Pestov.
\newblock {Characterization of the inverse problem data for one-dimensional
two-velocity dynamical system}.
\newblock {\em St Petersburg Mathematical Journal} 06/2015; 26(3):411-440.
DOI:10.1090/S1061-0022-2015-01344-7.

\bibitem{B EACM}
M.I.Belishev.
\newblock {Boundary Control Method.}
\newblock {\em Encyclopedia of Applied and Computational
Mathematics}, Volume no: 1, Pages: 142--146. DOI:
10.1007/978-3-540-70529-1. ISBN 978-3-540-70528-4

\bibitem{BV Rendiconti}
M.I.Belishev, A.F.Vakulenko.
\newblock {On characterization of inverse data in the boundary control method.}
\newblock {\em Rend. Istit. Mat. Univ. Trieste},  Volume 48 (2016), 49--77.
DOI: 10.13137/2464-8728/13151.

\bibitem{B UMN}
M.I.Belishev. {\newblock Boundary Control and Tomography of
Riemannian Manifolds.}
\newblock {\em Russian Mathematical Surveys}, 2017, 72:4, 581--644.
doi.org/10.4213/rm 9768.

\bibitem{BBlag book}
M.I.Belishev, A.S.Blagoveschenskii.
\newblock {Dynamical Inverse Problems of Wave Theory.}
\newblock {SPb State University, St-Petersburg}, 1999
(in Russian).

\bibitem{BMikh Dirac IP}
M.I.Belishev, V.S.Mikhaylov.
\newblock{Inverse problem for one-dimensional
dynamical Dirac system (BC-method)}.
\newblock{\em Inverse Problems}, Vol. 30, no. 12, doi:10.1088/0266-5611/30/12/125013,
2014.

\bibitem{Blag_71}
A.S.Blagovestchenskii.
\newblock{On a local approach to the solving the
dynamical inverse problem for inhomogeneous string.}
\newblock{\em Trudy MIAN V.A.~Steklova } 115 (1971), 28-38 (in Russian).

\bibitem{Blag2}
A.S.Blagovestchenskii.
\newblock{Inverse Problems of Wave Processes.}
\newblock{\em ZSP, Netherlands}, 2001.


\bibitem{GohbKrein}
I.Ts.Gohberg, M.G.Krein.
\newblock{Theory and Applications of Volterra Operators in Hilbert Space.}
\newblock{\em Transl. of Monographs No. 24, Amer. Math. Soc},
\newblock{Providence. Rhode Island}, 1970.

 \end{thebibliography}
\end{document}